\documentclass{article}

\usepackage{amsmath}
\usepackage{amsthm}
\usepackage{url}
\usepackage{enumerate}
\usepackage{graphicx}

\newtheorem{thm}{Theorem}

\newtheorem{cor}[thm]{Corollary}

\newtheorem{claim}[thm]{Claim}

\begin{document}

\title{Blind-friendly von Neumann's Heads or Tails}
\author{Vin{\'i}cius G. Pereira de S{\'a} and Celina M. H. de Figueiredo}
\date{}
\maketitle

\begin{abstract}
The toss of a coin is usually regarded as the epitome of randomness, and has been used for ages as a means to resolve disputes in a simple, fair way. Perhaps as ancient as consulting objects such as coins and dice is the art of maliciously biasing them in order to unbalance their outcomes. However, it is possible to employ a biased device to produce equiprobable results in a number of ways, the most famous of which is the method suggested by von Neumann back in 1951. This paper addresses how to extract uniformly distributed bits of information from a nonuniform source. We study some probabilities related to biased dice and coins, culminating in an interesting variation of von Neumann's mechanism that can be employed in a more restricted setting where the actual results of the coin tosses are not known to the contestants.
\end{abstract}

\section{Introduction.}\label{s:intro}

Estimating probabilities is one of those tasks at which the human brain seems to be not very good.
Conditional probabilities, in particular, are frequently defying---and defeating---one's intuition, from the uninitiated to the specialist. This explains to a certain extent why people lose money gambling and on similar activities. Led astray by instinct, the inadvertent player overlooks probabilistic subtleties and misestimates the odds.

In this paper, we study some probabilities which may be rather counterintuitive. 
Although our results were formulated in the context of games between two opponents, they could as well have been framed in terms of the general engineering problem of converting biased randomness into unbiased randomness, 
a fascinating subject with plenty of literature available (see, for instance,~\cite{Chor85, Elias72, Samuelson68, Srinivasan99, Vembu95, Zuckerman91}). 

In Section~\ref{s:game},
we describe a game where the winning chances depend on the fairness of a die with $n \geq 2$ sides in a nontrivial way. Indeed, even knowing beforehand the exact probability associated to each side of the die, it may not be so easy to decide, between two seemingly equivalent strategies, which one is the most advantageous. Still, it is possible to show that one strategy is never worse than the other, no matter the bias or the number of sides of the die. 
We then look at some special cases. 
Particularly for $n=2$, 
we derive the winning chances 
as a function of the bias, showing how to maximize 
the probability that a given player wins.

In Section~\ref{s:negative_knowledge},
we consider the following questions regarding three independent and identically distributed (iid) random variables $A,B$, and $C$.
Are the events $C = B$ and $B \neq A$ always independent?
When, and to which extent, may (the knowledge of) the former affect (the probability of) the latter?
The answers to these questions
are closely related to the results
in Section~\ref{s:game}.

Finally, in Section~\ref{s:blind}, we leverage the results from the previous sections into a variation of von Neumann's method of playing a fair heads or tails with a biased coin. Numerous improvements on von Neumann's original idea have been studied over the decades (\cite{Blum86, Hoeffding70, Peres92, Stout84}, to mention but a few). Our method handles, however, an extra restriction: the players are not able to know the actual result of each coin toss, instead they are only aware of whether each toss produced the same result as the previous one along a sequence of independent tosses.

\section{A triple or two straight doubles?}\label{s:game}

In this section, we discuss a very simple dice game, not only for its own sake, but also for the useful inequality we obtain from it.

Consider a die with $n \geq 2$ sides. 
We propose a game where each player must choose between two strategies: playing for a ``triple'' or two ``consecutive doubles''. In the former strategy, the player throws the die three times, scoring a point if the die produces three identical results. In the latter, the player throws the die four times, scoring a point if the results after the first and the second throws match one another \emph{and} the third and fourth results also match one another. After repeating their sequences of throws a previously arranged number of times, the player with more points wins the game.

One might think that both strategies are equally good (or equally terrible, depending on how big $n$ is).
Indeed, when playing for a triple, there are two critical moments at which the player needs to be lucky: at the second throw, whose result is required to match the first one; and at the third throw, which is also required to match the other two. When playing for two doubles in a row, there are \emph{also} two such critical moments: at the second throw, when the player wants the die to match the result that just preceded it, and again at the fourth throw, for the same reason. 
In spite of their apparent equivalence, 
the odds for both strategies are not always as good, and how more
advantageous one strategy is depends on the fairness of the die.

In order to get some intuition, suppose our die is strongly biased towards one particular outcome (say it results in a $1$ in $90\%$ of the throws). What difference does it make, now, to play for a triple or two straight doubles? When playing for a triple, one has only three chances of ``going wrong'', that is, of throwing the die and obtaining something other than $1$. (Of course it is possible that such other-than-$1$ outcome appears three times in a row, but that is definitely unlikely.) On the other hand, when playing for two straight doubles, one has in practice \emph{four} chances of ``going wrong'', since any other-than-$1$ outcome that ends up occurring will most likely remain unmatched. Now we proceed to a more formal discussion.

\begin{claim}\label{cl:4dice}
If $A,B,C$, and $D$ are iid discrete random variables, then $\textrm{Pr}\{A=B~\land~C=D\} \leq \textrm{Pr}\{A=B=C\}.$
\end{claim}
\begin{proof}
Let $\Omega$ be the range of possible values of $A,B,C$, and $D$. For $i \in \Omega$, let $p_i$ be the probability $\textrm{Pr}\{A=i\} = \textrm{Pr}\{B=i\} = \textrm{Pr}\{C=i\} = \textrm{Pr}\{D=i\}$ that a given variable takes value $i$. Let $E_{2,2}$ denote the event that $A=B$ and $C=D$, and let $E_3$ denote the event that $A=B=C$. Since the variables are independent, we can write
\begin{eqnarray}\label{eq:E22}
\textrm{Pr}\{E_{2,2}\} = \left(\sum_{i \in \Omega} p_i^2\right)^2,
\end{eqnarray}
and
\begin{eqnarray}\label{eq:E3}
\textrm{Pr}\{E_3\} = \sum_{i\in \Omega} p_i^3.
\end{eqnarray}
We now show that $\textrm{Pr}\{E_{2,2}\} \leq \textrm{Pr}\{E_3\}$. Inspired by its use in~\cite{Guilherme}, we recall the Cauchy inequality~\cite[p.~373]{Cauchy}
$$\left(\sum_{i\in \Omega} x_i y_i\right)^2 \leq \left(\sum_{i \in \Omega} x_i^2\right)\left(\sum_{i\in \Omega} y_i^2\right).$$
Setting $x_i = p_i^{3/2}$ and $y_i = p_i^{1/2}$, we obtain
$$\left(\sum_{i\in \Omega} p_i^2\right)^2 \leq \left(\sum_{i \in \Omega} p_i^3\right)\left(\sum_{i\in \Omega} p_i\right) = \sum_{i\in \Omega} p_i^3 $$
as desired.
\end{proof}

As a consequence, the probability of obtaining a triple in three throws of a die is never less than the probability of obtaining two straight doubles in four throws of that die. It should now be clear that equality must not be taken for granted. We look at some special cases.

From now on, let $A,B,C$, and $D$ be specifically the results of four consecutive throws of an $n$-sided die; that is, iid random variables defined by selections on the sample space $\Omega = \{1,\ldots,n\}$. We denote by $p_i$ the probability that $i$ is the result obtained by throwing that die, for $i \in \Omega$. 

\paragraph{Perfectly fair dice.}
Consider the case where $p_i = 1/n$ for all $i \in \Omega$; that is, our die is ``perfectly fair''. In this situation, equations (\ref{eq:E22}) and (\ref{eq:E3}) yield
$$\textrm{Pr}\{E_3\} = \textrm{Pr}\{E_{2,2}\} = \frac{1}{n^2},$$
and therefore both players have the same probability of winning the game. 
As a matter of fact, it is a simple exercise to show that Claim~\ref{cl:4dice} holds with equality if, and only if, 
either the random variables are uniformly distributed over a subset of their sample space  
or one of the possible outcomes occur with probability $1$.

\paragraph{Perfectly loaded dice.}
Suppose the die was manufactured in such a way that the same side always ends upwards after a throw, e.g.~$p_1 = 1$, and $p_i = 0$ for $i \in \{2,\ldots,n\}$. Of course both strategies will score a point at each and every sequence of throws, and the game will almost certainly end in a draw. As in the perfectly fair case, there is no preferable strategy.

\paragraph{Coins.}
We now focus on the case $n=2$ and call our die a coin. We are interested in the function $d:\left[0,1\right] \to \left[0,1\right]$ that quantifies the advantage $d(p) = \textrm{Pr}\{E_3\} - \textrm{Pr}\{E_{2,2}\}$ of playing for a triple when the probability of our coin landing heads is $p$. By making $p_1 = p$ and $p_2 = 1-p$ in (\ref{eq:E22}) and (\ref{eq:E3}), we obtain
\begin{eqnarray*}
d(p)&=&\left[p^3 + (1-p)^3\right] - \left[p^2 + (1-p)^2\right]^2 \\
       &=&-4p^4 + 8p^3 - 5p^2 + p.
\end{eqnarray*}
It is now easy to see (Figure~\ref{f:corcovas}) that $d$ has minimum value $0$ at $p \in \{0,0{.}5,1\}$ as expected, and maximum value $0{.}0625$ at $p = 1/2 \pm  1/(2\sqrt{2}) \approx 0{.}5 \pm 0{.}3536$. In other words, $6{.}25\%$ is the largest possible probabilistic advantage that can be obtained in this game\footnote{
Should a winning margin of $6{.}25\%$ appear to be
rather small, compare it with the casino's winning margin of $5{.}26\%$ in
the American roulette, and of only $2{.}70\%$ in the European roulette.
Yet, we do not see casinos going bankrupt quite too often!
}. 
This is achieved by the player who plays for a triple when the probability of the coin landing heads (or tails, by symmetry) is approximately $85.36\%$. 

\begin{figure}[h]
\vspace{0.4cm}
\centering
\includegraphics[height=45mm]{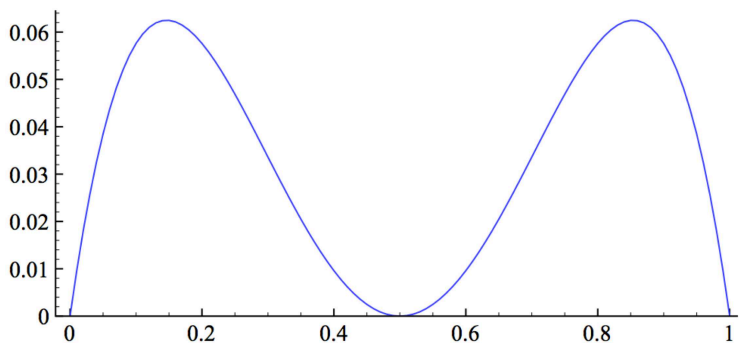}
\caption{\label{f:corcovas}Advantage of playing for a triple as a function of $p$}
\end{figure}

\section{Seemingly irrelevant knowledge.} \label{s:negative_knowledge}

Let us look at a different experiment. Someone throws a die with $n \geq 2$ sides three times in a row, and writes down the sequence of results, calling them $A, B$, and $C$. You want to guess whether $C = B$. Does it make any difference if you are told that $B \neq A$? 

Because the throws of the die are mutually independent, 
it is tempting to say
that the fact that the first two results ($A$ and $B$) do not match has no correlation whatsoever to whether the \emph{third} result ($C$) matches the second one. However, 
this reasoning turns
out to be deceiving.
The \emph{variables} are independent, but are the events they define necessarily so? What if the die is not perfectly fair?

Here again we start with an intuitive discussion. 
Suppose the die is so biased that one of its sides (say, $1$) lands upwards in $90\%$ of the throws. 
In this case, the event $B \neq A$ reveals that something unexpected happened, that is, $A$ and $B$ are not both $1$.
Since it is highly probable that $C$ \emph{is} $1$, 
we cannot anymore say that our
assessment of the probability associated to the event $C \neq B$ was unaffected by the knowledge that $B \neq A$. 

More formally, 
let $\Omega$ be the range of possible values of three iid random variables $A,B$, and $C$. For $i \in \Omega$, let $p_i$ be the probability $\textrm{Pr}\{A=i\} = \textrm{Pr}\{B=i\} = \textrm{Pr}\{C=i\}$ that a given variable takes value $i$.

Clearly,
\begin{eqnarray}\label{eq:PrCA}
\textrm{Pr}\{C=B\} = \sum_{i \in \Omega} p_i^2
\end{eqnarray}
and
\begin{eqnarray}\label{eq:PrCAnotB}
\textrm{Pr}\{C=B~|~B\neq A\} =  \frac{\textrm{Pr}\{C=B\neq A\}}{\textrm{Pr}\{B\neq A\}} 
                                               = \frac{\sum_{i \in \Omega} [p_i^2 (1-p_i)]}{\sum_{i \in \Omega} [p_i (1-p_i)]}.
\end{eqnarray}

The following result relates the two probabilities above.

\begin{claim}[Fonseca et al.~\cite{Guilherme}, Lemma $6$]\label{cl:3dice}
Given three iid discrete random variables $A,B$, and $C$, we have $\textrm{Pr}\{C=B~|~B\neq A\} \leq \textrm{Pr}\{C=B\}.$
\end{claim}

For an example where equality does not hold, let \mbox{$\Omega = \{1,2,3\}$}, \mbox{$p_1 = 0.8$}, \mbox{$p_2=p_3=0.1$}. In this case, $\textrm{Pr}\{C=B\} = 0.66$, whereas $\textrm{Pr}\{C=B~|~B\neq A\} \approx 0.429$.

In the realm of dice, the probability that the third result matches the second one given that the second result does not match the first one in a sequence of three independent throws of a die is less than or equal to the unconditional probability that the third result matches the second one. 

\paragraph{Perfectly fair dice.}
By replacing $p_i$ with $1/n$ in equations~(\ref{eq:PrCA}) and~(\ref{eq:PrCAnotB}), it is straightforward to verify 
that both probabilities are equal to $1/n$. In other words, if the die is fair, then the probability of having a match between the third and the second results is not at all affected by the fact that the second throw and the first one do not match. 
The events are in this case independent.

\paragraph{Perfectly loaded dice.}
If the die always lands with the same side upwards, 
then it is impossible that two throws have different
outcomes. Our comparison here is then meaningless.

\paragraph{Coins.}
When $n=2$, the unconditional probability that two tosses of the coin yield the same result is clearly 
\begin{eqnarray}\label{eq:pure2coins}
\textrm{Pr}\{C=B\}  =  p_1^2 + p_2^2 = p^2 + (1-p)^2 = 2p^2 - 2p + 1,
\end{eqnarray}
where $p$ denotes, without loss of generality, the probability of heads. 
Such a probability has minimum value $0{.}5$ at $p = 0{.}5$ (i.e.~for a fair coin).

On the other hand, the conditional probability we are interested in 
is obtained from equation~(\ref{eq:PrCAnotB}) and translates to 
\begin{eqnarray*}
\textrm{Pr}\{C=B~|~B\neq A\}  
                                               & = & \frac{p^2 (1-p) + (1-p)^2 p}{2p - 2p^2} \\
                                               & = & \frac{p - p^2}{2p - 2p^2} \\
                                               & = & 0.5,
\end{eqnarray*}                                               
regardless of $p$ (of course, provided $p$ is not $0$ or $1$). 

A simpler way of seeing this is noting that, 
given that the first two coins came up differently,
the second coin is uniformly distributed over heads and tails. 
Thus the probability that the third coin results the same as the second is
$0.5p + 0.5(1-p) = 0.5.$

This interesting result, which motivates the whole next section, is summarized in the corollary below. We recall that a Bernoulli random variable is such that it takes value $1$ with ``success probability'' $p$ and value $0$ with ``failure probability'' $1-p$.

\begin{cor}\label{co:fair}
Given three independent Bernoulli random variables $A,B$, and $C$ with success probability $0 < p < 1$, we have $\textrm{Pr}\{C=B~|~B\neq A\} = 0{.}5$ regardless of $p$.
\end{cor}

\section{Hand claps and whistles.}\label{s:blind}

The idea of playing a fair heads or tails with a biased coin\footnote{
In spite of the numerous references to such entity over the centuries, some recent evidence seems to show that there is no such thing as a biased coin that is caught by the hand (i.e.~the coin is allowed to spin in the air, but not to bounce). See \cite{Gelman02} for finding out why the biased coin may be considered the unicorn of probability theory.
}
is attributed to von Neumann~\cite{vonNeumann51}.
The coin is tossed twice in a row. The first player wins if the outcome is a heads-tails sequence, whereas the second player wins with a tails-heads sequence. If two identical results are obtained, another \emph{turn} of two coin tosses starts from scratch. The probability of winning the game at a certain turn is identical for both players, namely $p(1-p)$, where $p$ is the probability that a coin toss results heads. Thus, the probability that \emph{some} player wins at a certain turn is $q = 2p(1-p)$, and the number of turns until the game has a winner is a geometric random variable $X$ whose expectation is $\textrm{E}\{X\} = 1/q = 1/(2(p-p^2))$. Since each turn comprises exactly two coin tosses, the expected number of tosses until a player wins using von Neumann's method is $2 \textrm{E}\{X\} = 1/(p-p^2)$.

Now suppose we have a situation where the coin is concealed, that is, the contestants are not able to see the outcome of each toss. Instead, they can only figure out, after each toss, whether the result happened to match the one that just preceded it, which we refer to as the \emph{base} result. 
For the sake of illustration, consider that the coin is tossed by a
trusted third party who claps hands each time the coin toss results the
same as in the previous toss, and who whistles otherwise.
Only after the very first toss (the \emph{initial toss}, in each game), no sound is produced, since there is no base result to compare with. 
We want to assure fairness, and we certainly cannot use von Neumann's original method, since the actual coin results are not known to the players. 

We first discuss wrong ways of trying to
obtain a fair game.

\subsection{Unsuccessful attempts at fairness.}

In the proposed setting, the players cannot see the actual results,  
but they hear something after each toss\footnote{Except for the initial toss, as mentioned above.}: a hand clap (Cl) or a whistle (Wh). Then why not simply regard a hand clap as heads, a whistle as tails, and play the good old heads or tails? 
In other words, after the initial toss, have the coin be tossed exactly once more. Player $1$ wins with a hand clap, Player $2$ wins with a whistle. Of course one can play such a game, but one should not expect it to be fair if the coin itself is not fair. Indeed, the greater the coin bias, the greater the advantage of playing for a hand clap. Formally, if $p$ is the probability that the coin lands heads, then the probability that Player $1$ wins is equal to the probability that two independent tosses of the coin yield the same result, namely $p^2 + (1-p)^2 = 2p^2 - 2p + 1$, 
whose minimum value $0{.}5$ is obtained at $p=0{.}5$, as in~(\ref{eq:pure2coins}).

Aiming at neutralizing the coin bias, a second, natural attempt is to use a straightforward translation of von Neumann's idea based on the perceived sounds. After the initial toss, which sets the first base value, the coin is tossed twice in a row. Player $1$ wins with a Cl-Wh sequence,
Player $2$ wins with a Wh-Cl sequence.
If any other sequence occurs, the game proceeds with another turn of two coin tosses. Is the game now fair, regardless of the coin bias? We show that it is not, except again if $p = 0{.}5$.

To calculate the probability $P_1$ that Player $1$ wins the game, 
let $F^H$ and $F^T$ denote the events where the initial toss results heads (H) or tails (T), respectively,  so that 
\begin{eqnarray*}
P_1  =  \textrm{Pr}\{P_1~|~F^H\} \cdot \textrm{Pr}\{F^H\} + \textrm{Pr}\{P_1~|~F^T\} \cdot \textrm{Pr}\{F^T\}
\end{eqnarray*}
by the law of total probability.
By letting  
$P_1^H = \textrm{Pr}\{P_1~|~F^H\}$ and
$P_1^T = \textrm{Pr}\{P_1~|~F^T\}$,
we can write
\begin{eqnarray}\label{eq:p1}
   P_1   =  P_1^H \cdot p + P_1^T \cdot (1-p).
\end{eqnarray}

We now obtain $P_1^H$ conditioned on the results of the two coin tosses that followed the initial toss (whose result was heads, by definition):
\begin{itemize}
\item {[base H]} H-H --- this sequence yields two hand claps, no player wins yet, and the probability that Player $1$ eventually wins is still $P_1^H$ due to the memorylessness of the process, which now starts over with a new turn of two tosses still having heads as base value;

\item {[base H]} H-T --- this gives a Cl-Wh sequence, and victory is awarded to Player $1$; 
\item {[base H]} T-H --- we have two whistles and no winner, hence the probability of Player $1$ winning is still $P_1^H$, as in the case H-H;

\item {[base H]} T-T --- a whistle followed by a hand clap: victory goes straight to Player $2$. 
\end{itemize}

Using again the law of total probability, we can write
\begin{eqnarray*}
P_1^H & = & \sum_{S\in{\cal S}} \textrm{Pr}\{P_1^H~|~S\} \cdot \textrm{Pr}\{S\}\\
      & = & P_1^H  \cdot p^2
       +  1      \cdot p(1-p)
       +  P_1^H  \cdot (1-p)p
       +  0      \cdot (1-p)^2,
\end{eqnarray*}

\noindent where ${\cal S} = \{$H-H, H-T, T-H, T-T$\}$ is the set of events associated to the two coin tosses that follow the initial toss, as mentioned above. After some easy manipulations, we obtain 
\begin{equation}\label{p1h}
P_1^H = p.
\end{equation}

\noindent An analogous argument shows that
\begin{equation}\label{p1t}
P_1^T = 1-p.
\end{equation}

By plugging (\ref{p1h}) and (\ref{p1t}) into (\ref{eq:p1}), we obtain $P_1 = p^2 + (1-p)^2 = 2p^2 - 2p + 1$. 
This is again the same expression as in (\ref{eq:pure2coins}), which has minimum value $0{.}5$ at $p = 0{.}5$. 
In short, the game is only fair when the coin itself is fair.

Before proceeding,
it is interesting to notice that here---unlike the coin game proposed in Section~\ref{s:game}, where a probability of heads around $0{.}5 \pm 0{.}3536$ would maximize a player's winning odds---the greater the coin bias, the greater the probability that Cl-Wh wins over Wh-Cl, in spite of the apparent symmetry of both strategies. 
One could try this game against an inadvertent opponent (perhaps someone who erroneously regard it 
as being equivalent to von Neumann's method). 
However, using a similar reasoning as in the calculation of $P_1$, we see that
the expected number of tosses before the game finishes (disregarding the initial toss) is $1/(p-p^2)$,
as in von Neumann's method. 
This fact cannot be
overlooked. For example, when the coin is so biased that $p = 0{.}999$,
though the Cl-Wh player has a winning probability of $98{.}02\%$, the coin
needs to be tossed $1001$ times on average before the game has a winner!

In the next section, we look into a variation 
that successfully shields the players from whatever bias the concealed coin might have (provided both sides appear with non-zero probability).

\subsection{Fair game with a concealed biased coin.}\label{ss:fairht}

We want to devise a coin game with the following properties:

\begin{enumerate}[1)]
\item the coin is possibly biased, yet the exact bias is unknown to the players;
\item both players have equal chances of winning;
\item the actual coin results are not revealed, but it is possible to infer whether any two consecutive coin tosses yielded the same result.
\end{enumerate}

The two first conditions are the usual ones. Ideally, we would like to cope with the third, new condition without incurring any increase in the 
expected duration of the game.

First, let us consider the natural solution of tossing the coin four times, considering the sounds $X$ and $Y$ produced after the second and fourth tosses, respectively. Player $1$ wins with a Cl-Wh sequence, Player $2$ wins with Wh-Cl, and the other possible sequences are discarded, triggering four fresh tosses of the coin. By doing so, $X$ and $Y$ are clearly independent, and this approach corresponds exactly to von Neumann's idea, therefore assuring a fair game. The drawback is that it may take twice as long, with an expected number of $1 / (-2p^4 + 4p^3 - 3p^2 + p)$ coin tosses, which can be easily checked.

Now consider the following variation. The main action consists of tossing the coin two consecutive times, whose results we call respectively $A$ and $B$, constituting a \emph{turn}. If a whistle is heard after the second toss (meaning $B \neq A$), then a third coin toss---whose result we call $C$---will decide the game. Player $1$ wins if $C \neq B$ (indicated by the subsequent whistle), Player $2$ wins if $C = B$ (indicated by the subsequent hand clap).
On the other hand, if the sound alert after the second coin toss is a hand clap (meaning $B = A$), then the game continues with a fresh turn of two coin tosses that will set the values of $A$ and $B$ from scratch, configuring a perfectly memoryless process. 
Note that an initial toss that would set a base value prior to the first turn is not even necessary here, since the players only care about sound alerts that come with even parity (after the second, fourth, sixth etc.~tosses) until a whistle with even parity is heard for the first time. 
When it eventually happens, a single, decisive extra toss closes the game.

\begin{thm}
The proposed coin game, where Player $1$ wins with a whistle of even parity immediately followed by another whistle, and Player $2$ wins with a whistle of even parity immediately followed by a hand clap, is a perfectly fair game, no matter the probability of heads $0 < p < 1$ of the employed coin.
\end{thm}
\begin{proof}
Let $A$ and $B$ be the results of the two coin tosses that preceded the first whistle with even parity. That very whistle indicates $B \neq A$. The next coin toss, whose result we call $C$, will decide the game in Player 1's favor if $C \neq B$, and in Player 2's favor if $C = B$. Using a simple (tails $\to 0$, heads $\to 1$) mapping, those three results are clearly independent Bernoulli random variables with success probability $p$.
Thus, 
by Corollary~\ref{co:fair}, both players have identical probabilities $$\textrm{Pr}\{C \neq B~|~B\neq A\} = \textrm{Pr}\{C = B~|~B\neq A\} = 0{.}5$$ of winning the game.
\end{proof}
 
As for the expected number of tosses, notice that the event that a whistle is heard after the second toss of a turn corresponds to the event that the two consecutive tosses of that turn produces either a H-T or a T-H sequence, hence it occurs with probability $2p(1-p)$. The number of turns before the game finishes is therefore a geometric random variable $X$, whose expectation is $\textrm{E}\{X\} =1/(2p(1-p))$, plus one (corresponding to the final, decisive toss). Thus, the expected number of tosses in the game is 
\begin{equation}\label{expected_number_of_tosses}
2 \cdot \textrm{E}\{X\} + 1 = \frac{1}{p-p^2} + 1,
\end{equation}
which is only one coin toss greater than in the original von Neumann's method.
 
It is interesting
to observe that, if in an attempt to reduce its expected duration we
decided the game right after the first whistle regardless of its
parity, then two unfortunate consequences would ensue:
\begin{enumerate}[(i)]
\item the intended ``improvement'' would \emph{not} decrease the expected length of the game at all; and, most importantly, 
\item we would not be able to assure the fairness of the game anymore! 
\end{enumerate}

To show (i), let $F^H$ (and, respectively, $F^T$) denote the event that the first coin toss in the game yields heads (respectively, tails), and let $Y$ be the random variable corresponding to the number of tosses before the first whistle is heard. Using conditional expectations, we can write the overall expected number of tosses in the modified game as 
\begin{eqnarray*}
1 + \textrm{E}\{Y\} & = & 1 + \textrm{E}\{Y~|~F^H\} \cdot \textrm{Pr}\{F^H\} + \textrm{E}\{Y~|~F^T\} \cdot \textrm{Pr}\{F^T\} \\
                              & = & 1 + \left(1 + \frac{1}{1-p}\right) \cdot p + \left(1 + \frac{1}{p}\right) \cdot (1-p) \\ 
                              & = & 1 + \frac{1}{p - p^2},
\end{eqnarray*}  
which is exactly the same as in (\ref{expected_number_of_tosses}).

To show (ii), we argue that Player $1$---who wins the game if the sound after the first whistle is another whistle---will win the game exactly if:
\begin{itemize}
\item the first coin toss yields H, and the toss immediately after the first \linebreak occurrence of a T (which produces the first whistle of the game) yields~H (producing the second whistle in a row); or, analogously,
\item the first coin toss yields T, and the toss immediately after the first occurrence of a H yields T.
\end{itemize}
Thus, the probability that Player $1$ wins the game is $p^2 + (1-p)^2 = 2p^2 - 2p + 1$, the same function of $p$ seen in (\ref{eq:pure2coins})---which has minimum value $0{.}5$ exactly at $p=0{.}5$---yet again.

\section{Acknowledgements.}
The ``triple or two straight doubles'' question was proposed to us by Guilherme Dias da Fonseca when he was studying randomized data structures. We are also grateful to Alexandre Stauffer, Luiz Henrique de Figueiredo, and to the anonymous referees for the numerous valuable suggestions.

\bigskip

\noindent\textbf{Vin{\'i}cius Gusm{\~a}o Pereira de S{\'a}} received his D.Sc.~degree from Universidade Federal do Rio de Janeiro, where he is currently an associate professor. When he is neither working with combinatorics or programming computers, he can be seen bicycling or trying to play the piano.

\noindent\textit{Departamento de Ci{\^e}ncia da Computa{\c c}{\~a}o, UFRJ, Brazil\\
vigusmao@dcc.ufrj.br}

\bigskip

\noindent\textbf{Celina Miraglia Herrera de Figueiredo} is a full professor at the Systems Engineering and Computer Science Program of COPPE, Universidade Federal do Rio de Janeiro, where she has been a collaborator since 1991.

\noindent\textit{COPPE, UFRJ, Brazil\\
celina@cos.ufrj.br}


\begin{thebibliography}{10}

\bibitem{Blum86}
M. Blum, Independent unbiased coin flips from a correlated biased source: a finite state Markov chain, 
\emph{Combinatorica} \textbf{6} (1986) 97--108.

\bibitem{Cauchy}
A. Cauchy, \textit{Oeuvres 2, III} (1821).

\bibitem{Chor85}
B. Chor, O. Goldreich, Unbiased bits from sources of weak randomness and probabilistic communication complexity,
in \emph{Proc.~26th IEEE Symp.~on Foundations of Computer Science} (1985) 429--442.

\bibitem{Elias72}
P. Elias, The efficient construction of an unbiased random sequence, 
\emph{The Annals of Mathematical Statistics} \textbf{43} (1972) 865--870.

\bibitem{Gelman02}
A. Gelman, D. Nolan,
You can load a die, but you can't bias a coin, 
\emph{The American Statistician} \textbf{56} (2002) 308--311.

\bibitem{Guilherme}
G. D. da Fonseca, C. M. H. de Figueiredo, P. C. P. Carvalho,
Kinetic hanger, 
\emph{Information Processing Letters} \textbf{89} (2004) 151--157.

\bibitem{Hoeffding70}
W. Hoeffding, G. Simons, 
Unbiased coin tossing with a biased coin,
\emph{The Annals of Mathematical Statistics} \textbf{41} (1970) 341--352.

\bibitem{vonNeumann51}
J. von Neumann,
Various techniques used in connection with random digits,
\emph{J. Res. National Bureau of Standards} \textbf{12} (1951) 36--38.

\bibitem{Peres92}
Y. Peres,
Iterating von Neumann's procedure for extracting random bits,
\emph{The Annals of Statistics} \textbf{20} (1992) 590--597.

\bibitem{Samuelson68}
P. A. Samuelson,
Constructing an unbiased random sequence,
\emph{Journal of the American Statistical Association} \textbf{63} (1968) 1526--1527.

\bibitem{Srinivasan99}
A. Srinivasan, D. Zuckerman, 
Computing with very weak random sources,
\emph{SIAM Journal on Computing} \textbf{28} (1999) 1433--1459.

\bibitem{Stout84}
Q. F. Stout, B. Warren, 
Tree algorithms for unbiased coin tossing with a biased coin,
\emph{The Annals of Probability} \textbf{12} (1984) 212--222. 

\bibitem{Vembu95}
S. Vembu, S. Verd\'u,
Generating random bits from an arbitrary source: fundamental limits,
\emph{IEEE Transactions on Information Theory} \textbf{41} (1995) 1322--1332.

\bibitem{Zuckerman91}
D. Zuckerman,
Simulating BPP using a general weak random source,
in \emph{Proc.~32nd IEEE Symp.~on Foundations of Computer Science} (1991): 79--89.
\end{thebibliography}
\end{document}